\documentclass[8pt,onecolumn]{IEEEtran}
\usepackage{amsmath}
\usepackage{amssymb}
\usepackage{graphics}
\usepackage{epstopdf}
\usepackage{graphicx}
\usepackage{epsfig,amssymb}
\usepackage{amsmath}
\usepackage{mathrsfs}
\usepackage{color}
\usepackage{array}
\usepackage{amsfonts,booktabs}
\usepackage{lipsum,amsmath,multicol}
\usepackage{graphicx}
\usepackage{times}
\usepackage{subcaption}

\usepackage{pdfsync}
\usepackage{pgfplots}
\usepackage{pgfplotstable}
\usepackage[subnum]{cases}
\usepackage{multirow}
\usepackage{algorithm}
\usepackage{algpseudocode}
\usepackage{float}
\usepackage{setspace}
\usepackage{bm}

\usepackage{amsthm}
\usepackage{epstopdf}
\definecolor{blue}{rgb}{0,0,1}  
\newtheorem{theorem}{Theorem}
\newtheorem{lemma}{Lemma}

\newtheorem{remark}{Remark}

\newcommand{\beq}{\begin{equation}}
\newcommand{\eeq}{\end{equation}}
\newcommand{\beqa}{\begin{eqnarray}}
\newcommand{\eeqa}{\end{eqnarray}}

\newcommand{\paren}[1]{\left(#1\right)}

\newcommand{\field}[1]{\ensuremath{\mathbb{#1}}}
\newcommand{\R}{\ensuremath{\field{R}}} 
\newcommand{\I}[1]{\ensuremath{\mathsf{1}_{\left\{#1\right\}}}} 

\begin{document}

	\title{Optimal Privacy-Aware Dynamic Estimation}
 \author{Chuanghong Weng, Ehsan Nekouei, and Karl H. Johansson
\thanks{
	C. Weng and E. Nekouei are with the Department of Electrical Engineering, City University of Hong Kong. E-mail: {\tt \{cweng7-c,enekouei\}@cityu.edu.hk}. K. H. Johansson is with the School of Electrical Engineering and Computer Science, KTH Royal Institute of Technology. He is also affiliated with Digital Futures. E-mail: {\tt kallej@kth.se}.

The work was partially supported by the Research Grants Council of Hong Kong under Project CityU 21208921, a grant from Chow Sang Sang Group Research Fund sponsored by Chow Sang Sang Holdings International Limited, the Knut and Alice Wallenberg Foundation, the
Swedish Foundation for Strategic Research, and the Swedish Research Council.}}
	
	\maketitle
 \thispagestyle{empty}

\pagestyle{empty}

\begin{abstract}
In this paper, we develop an information-theoretic framework for the optimal privacy-aware estimation of the states of a (linear or nonlinear) system. In our setup, a private process, modeled as a first-order Markov chain, derives the states of the system, and the state estimates are shared with an untrusted party who might attempt to infer the private process based on the state estimates. As the privacy metric, we use the mutual information between the private process and the state estimates. We first show that the privacy-aware estimation is a closed-loop control problem wherein the estimator controls the belief of the adversary about the private process. We also derive the Bellman optimality principle for the optimal privacy-aware estimation problem, which is used to study the structural properties of the optimal estimator. We next develop a policy gradient algorithm, for computing an optimal estimation policy, based on a novel variational formulation of the mutual information. We finally study the performance of the optimal estimator in a building automation application.
\end{abstract}

\section{Introduction}
\subsection{Motivation}
Privacy is a major concern for the users of networked control systems (NCSs), as these systems increasingly rely on third-party computing entities, \emph{e.g.,} cloud computing units. However, these computing entities may act as an honest-but-curious adversary who has lawful access to the sensor measurements of a system but may attempt to infer private information from the measurements.  For instance, the occupancy of a building, which is highly sensitive, can be accurately inferred from its CO$_2$ measurements.

Privacy breaches may have serious consequences for the designers and operators of NCSs, such as smart buildings and intelligent transportation systems. Thus, it is necessary to address the privacy concerns during the design process. In this paper, we will develop an information-theoretic framework for the optimal privacy-aware design of state estimators for NCSs.

\subsection{Contributions}
We study the optimal privacy-aware estimation for a system, as shown in Fig. \ref{Fig.EstInfSys}. In our setup, a private process drives the states of the system, and the state estimates are shared with an adversary who might use the state estimates to infer the private process. We consider the class of randomized estimators in which an estimator randomly selects a value as its output according to a probability distribution. In our formulation, the private process is modeled as a first-order Markov chain, and system dynamics can be either linear or nonlinear. 
\begin{figure}[H]
	\centering
	\includegraphics[width=0.6\textwidth]{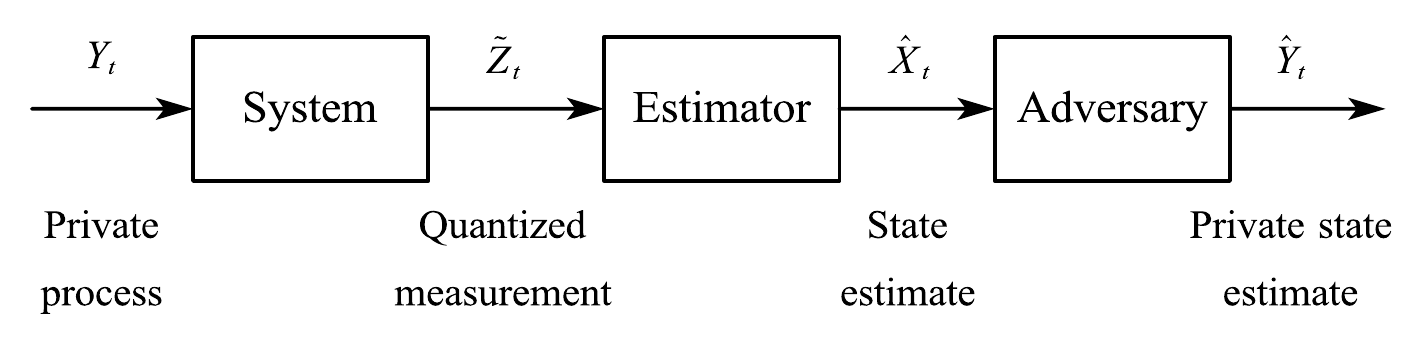}
	\caption{The privacy-aware state estimation setup. }
	\label{Fig.EstInfSys}
\end{figure}
We cast the optimal privacy-aware design of the estimator as an optimization problem wherein the objective is to minimize a linear combination of the estimation loss and a penalty on the leakage of private information via the estimator, captured by the mutual information between the private process and the state estimates. We first show that the optimal privacy-aware estimation is, in fact, an optimal control problem wherein the estimator controls the belief of the adversary about the private process.  Thus, different from the classical estimation, the privacy-aware estimation becomes a closed-loop control problem.
We derive the Bellman optimality principle associated with the optimal privacy-aware estimation problem and show that the optimal estimator consists of a forward filtering equation that computes the adversary's belief.

We next develop a gradient-based algorithm for computing an optimal estimation policy. To this end, we first derive an expression for the gradient of the objective function of the estimator design problem. We then develop a low-complexity variational solution for computing the log-likelihood functions that arise in the gradient expression.  We finally use a building automation application to demonstrate the effectiveness of our framework in ensuring privacy, and compare its performance with the additive noise approach to privacy.

\subsection{Related work}

The information-theoretic approach to privacy in dynamic settings was studied in the literature.  Li \emph{et al.} considered preserving the privacy of a household's electricity demand using a rechargeable electricity storage device in \cite{li2018information}. They studied the optimal charging policy of the storage device using mutual information as the privacy metric. 
The privacy filter design problem for a fully observable Markov process was studied in  \cite{erdemir2020privacy} using mutual information as the privacy metric. The authors in \cite{mochaourab2018private} and \cite{cavarec2021designing} considered the privacy filter design for hidden Markov models using conditional entropy and mutual information, respectively, as privacy metrics.   They proposed greedy and receding horizon solutions for the design of privacy filters. In \cite{molloy2023smoother}, the authors studied the controller design problem for a partially observed Markov decision process to hinder the estimation of its state by an adversary. In their framework, the inference capability of the adversary was captured by the conditional entropy of the state trajectory given measurements and controls. The privacy filter design problem for cloud-based control was studied in \cite{tanaka2017directed, nekouei2019information} using directed information as the privacy metric. 

We note that additive perturbation techniques have been used to ensure privacy in dynamic settings, \emph{e.g.,} see \cite{cortes2016differential} and references therein. The authors in \cite{le2013differentially} proposed a privacy-preserving filtering scheme based on differential privacy. In \cite{huang2014cost, mo2016privacy}, the authors developed additive noise mechanisms for protecting the privacy of the initial state in average consensus problems. The authors in \cite{kawano2020design} established a connection between the concept of differential privacy and the input observability of a control system. Gohari \emph{et al.} developed a policy synthesis algorithm for protecting the state transition probability in Markov decision process (MDP) using a Dirichlet mechanism with a bounded privacy protection cost. A linearly-solvable MDP solution for privacy-aware demand response in electricity grids was developed in \cite{hassan2021privacy}. The reward information protection for MDPs was studied in \cite{wang2019privacy} by adding functional noise to the value function. 

\subsection{Outline}
The remainder of the paper is organized as follows. Section \ref{Sec:Prob} presents our system model and problem formulation. Section \ref{Sec:Stru} investigates the structural properties of the optimal estimator. Section \ref{Sec:Comp} presents the policy gradient algorithm for the optimal estimator design.  Section \ref{Sec:Sim} presents the numerical results, followed by the concluding remarks in Section \ref{Sec:Con}.
\subsection{Notation}
 We denote random variables with uppercase letters, \emph{e.g.,} $\paren{X,Y}$, and denote their realization with lowercase letters, \emph{e.g.,} $\paren{x,y}$. In what follows, we use  $X^T$ as the shorthand for $\left[X_0,X_1,...,X_T \right]$.   We use $p\paren{X}$ and $P\paren{X}$ to represent probability density and probability mass functions, respectively. We also use  
	$p(\left. X\right |Y)$ and $P\paren{\left. X\right|Y}$ to represent the conditional density function and the probability probability mass function, respectively. 
 Given two random variables $X,Y$, $H\paren{\left. X \right|Y}$ and $I\paren{X;Y}$ denote the conditional entropy of $X$ given $Y$, and the mutual information between $X$ and $Y$, respectively. We use   
	$H\paren{X}$ to denote the discrete  entropy of $X$.
	 The Kullback–Leibler divergence (KL-divergence) between distributions $P\paren{\cdot}$ and $Q\paren{\cdot}$ is denoted by $D_{KL}\left[\left.P\paren{X}\right\|Q\paren{X}\right]$.

\section{System Model and Problem Formulation}\label{Sec:Prob}
Consider a  stochastic system in which the state of the system evolves according to the stochastic kernel
\begin{align}\label{Eq.Dynamics-state}
p(\left. X_{t+1}\right |X_t,Y_t)
\end{align}
where $X_t$ takes values 
 in $\mathcal{X}\subseteq\R^n$ is the system state at time $t$, $Y_t\in \mathcal{Y}$ is a discrete random variable encoding the private information at time $t$. The stochastic kernel in \eqref{Eq.Dynamics-state} specifies the conditional density of the state $X_{t+1}$ at time $t+1$  given $X_t$ and $Y_t$. We assume that the private process $\left\{Y_t\right\}_t$ is a first-order Markov process with the stochastic kernel $P\paren{\left. Y_{t+1} \right|Y_t}$, which specifies the conditional distribution of $Y_{t+1}$ given $Y_t$. We use $Z_t$ to denote the sensor measurements at time $t$, which takes values in $\mathcal{Z}\subseteq\R^r$ and evolves according to the stochastic kernel $p\paren{\left. Z_{t} \right|X_t}$. 
 
The above model allows us to study the privacy-aware estimation in nonlinear settings where the state and measurements evolve according to 
\begin{align}
	X_{t+1}&=f\left( X_t,Y_t,W_t \right)\nonumber\\
	Z_t&=g\left( X_t,V_t \right),\nonumber
\end{align}
where $W_t$ and $V_t$ represent the process noise and measurement noise, respectively. Another special case of the above model is the linear non-Gaussian system
\begin{align}
	X_{t+1}&=AX_t+BY_t+W_t\nonumber\\
	Z_t&=CX_t+V_t,\nonumber
\end{align}
where $V_t,W_t$ are non-Gaussian.

\subsection{Privacy-Aware State Estimation}
In this paper, we study the optimal design of an estimator for this system when the state estimates are shared with an untrusted adversary, as shown in Fig. \ref{Fig.EstInfSys}. Thus, the adversary may use the state estimates to infer the private process. In this subsection, we first present a motivating example that demonstrates that an adversary with access to the states of the system might be able to infer private information. We then present a framework for the optimal privacy-aware design of the state estimator. 
\subsubsection{Motivating Example}
Consider the evolution of CO$_2$ in a building, which can be modeled as 
\begin{align}\label{Eq.Dynamics-CO2}
X_{t+1}=aX_t+bY_t+W_t,
\end{align}
where  $X_t$ is the CO$_2$ level at time $t$, $Y_t$ is the occupancy, \emph{i.e.,} the number of people in the building, $W_t$ is the process noise, and $a,b\in\R$ are the system parameters. Also, let $Z_{t}=X_t+V_t$ denote the measurement of the CO$_2$ sensor at time $t$ where $V_t$ is the measurement noise. 

Occupancy of a building is highly private, and sensitive information, such as the mobility of the building occupants, can be inferred from the occupancy. However, the occupancy can be reliably inferred from the CO$_2$ measurement. To demonstrate this fact, we simulated the CO$_2$ level of a building using \eqref{Eq.Dynamics-CO2} with $a=0.75$,$b=0.2$. In our simulation,  $W_t$ and $V_t$ were modeled as independent and identically distributed zero mean Gaussian random variables with standard deviations $0.05$ and $0.1$, respectively. We then used a maximum likelihood estimator to infer the occupancy based on CO$_2$ measurements.

Figures  \ref{Fig.OccStateAdvA} and {\ref{Fig.OccStateAdvB}}  show a trajectory of CO$_2$ and its corresponding occupancy trajectory, respectively. In Fig. \ref{Fig.OccStateAdvC}, we plot the misdetection instances of the occupancy estimator when the CO$_2$ measurements in Fig. \ref{Fig.OccStateAdvA} are used to infer occupancy. In this figure, $\I{\hat{Y}_t\neq Y_t}$ is equal to one when the occupancy estimator makes a mistake and is zero otherwise. According to this figure, the maximum likelihood estimator can closely track the occupancy. Thus, sharing the CO$_2$ measurements with an untrusted party may result in loss of privacy for the building occupants.

\begin{figure}
\centering
  \begin{subfigure}[b]{.45\textwidth}
  \includegraphics[width=\linewidth,height=0.78\textwidth]{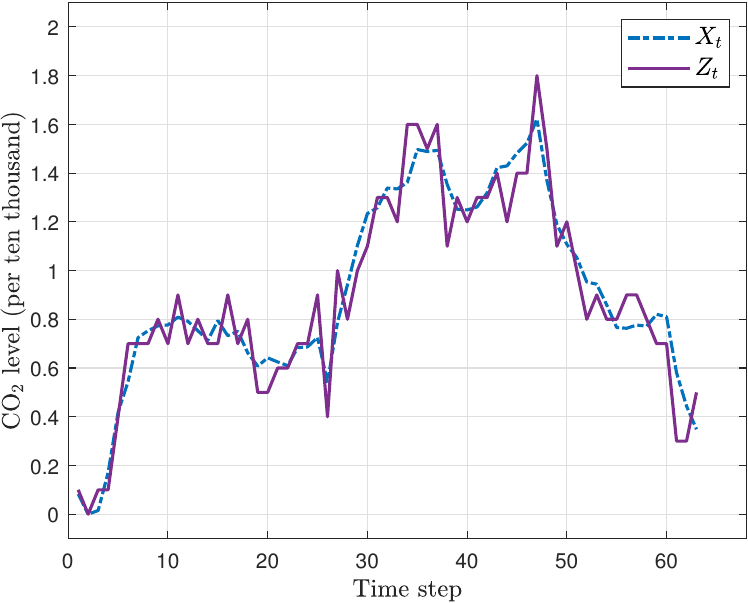}
    \caption{}
    \label{Fig.OccStateAdvA}
  \end{subfigure} \quad
  \begin{subfigure}[b]{.45\textwidth}
    \begin{subfigure}[b]{\textwidth}
      \includegraphics[width=\linewidth,height=0.35\textwidth]{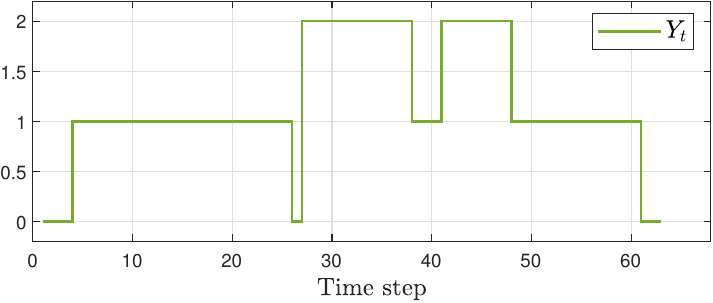}
      \caption{}
      \label{Fig.OccStateAdvB}
    \end{subfigure}
    \begin{subfigure}[b]{\textwidth}
      \includegraphics[width=\linewidth,height=0.35\textwidth]{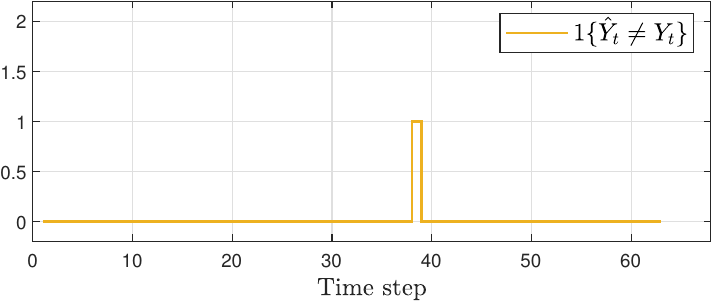}
      \caption{}
      \label{Fig.OccStateAdvC}
    \end{subfigure}
  \end{subfigure}
  \caption{Trajectories of CO$_2$ and sensor measurements ($a$), The occupancy trajectory ($b$), The misdetection instances of the occupancy estimator ($c$).}
  \label{Fig.OccStateAdv}
\end{figure}

\subsubsection{Optimal Design of Dynamic Estimators}
 We assume the estimator has access to the quantized version of sensor measurements. Let $B^{z}_{1},\dots,B^{z}_{N}$ denote a set of quantization cells, \emph{i.e.,}  $B^{z}_{i}$s are non-overlapping subsets of $\mathcal{Z}\subseteq\R^r$ with $\cup_iB^{z}_{i}=\mathcal{Z}$. Also, let $z_{\paren{i}}$ denote the cell center for $B^{z}_{i}$. We use $\tilde{Z}_t$ to denote the quantized version of $Z_t$, \emph{i.e.,} we have $\tilde{Z}_t=z_{(i)}$ when $Z_t$ is in $B^{z}_{i}$.

 We consider interval estimators where an estimator selects a region of $\mathcal{X}$ as its output. To this end, let  $B^{x}_{1},\dots,B^{x}_{M}$ denote a tessellation of $\mathcal{{X}}$, that is, $B^{x}_{i}$s are non-overlapping subsets of $\mathcal{{X}}\subseteq\R^n$ with $\cup_iB^{x}_{i}=\mathcal{{X}}$. We use ${x}_{(i)}$ to denote the cell center for $B^{x}_{i}$. 
 
Let $\hat{X}_t$ to denote the estimator's output at time $t$. A stochastic estimator randomly selects an element of $\left\{{x}_{(1)},\dots,{x}_{(M)}\right\}$ as its output based on the history of quantized measurements $\tilde{Z}^t=\left\{\tilde{Z}_0, \dots,\tilde{Z}_t\right\}$ and the history of its past decisions $\hat{X}^{t-1}=\left\{\hat{X}_{0},\dots,\hat{X}_{t-1}\right\}$. We use  $\pi_{t}\paren{\left. \hat{X_t}\right |\tilde{Z}^t,\hat{X}^{t-1}}$ to denote the estimation policy at time $t$, where $\pi_{t}\paren{\left. \cdot\right |\tilde{Z}^t,\hat{X}^{t-1}}$ is a  conditional probability distribution on $\left\{{x}_{(1)},\dots,{x}_{(M)}\right\}$, \emph{i.e.,} $\pi_{t}\paren{\left. {x}_{(i)}\right |\tilde{Z}^t,\hat{X}^{t-1}}$ is the probability that the estimator selects ${x}_{(i)}$ as its output.  We finally define a stochastic estimation policy over the horizon $t=0,\dots,T$ as the sequence of conditional probability distributions  $\left\{\pi_{t}\right\}^T_{t=0}$. 

 We use the mutual information between $\hat{X}^T=\left\{\hat{X}_0,\dots,\hat{X}_T\right\}$ and $Y^T=\left\{Y_0,\dots,Y_T\right\}$ as our privacy metric, which is defined as,
 \begin{equation}
     I\paren{\hat{X}^T;Y^T} = \sum_{\hat{x}^T,y^T}P\paren{\hat{x}^T,y^T} \log{\frac{P\paren{\hat{x}^T,y^T}}{P\paren{\hat{x}^T}P\paren{y^T}}}.
 \end{equation}

 Note that the mutual information captures the amount of common information between $\hat{X}^T$ and $Y^T$. When $ I \left(\hat{X}^{T};Y^{T}\right)$ is equal to zero, $\hat{X}^T$ and $Y^T$ will be statistically independent. However, the adversary can perfectly recover private information using the output of the estimator when $ I \left(\hat{X}^{T};Y^{T}\right)$  is equal to the entropy of $Y^{T}$ \cite{cover1999elements}. Thus, $ I \left(\hat{X}^{T};Y^{T}\right)$ captures the leakage level of private information through the estimator. Hence, a small value of $ I \left(\hat{X}^{T};Y^{T}\right)$ indicates a high privacy level.

The optimal privacy-aware dynamic stochastic estimation policy is the solution of the following optimization problem
\begin{equation}\label{Eq.OP}
\min_{\left\{\pi_{t}\right\}^T_{t=0}} L\left(\pi\right)= \min_{\left\{\pi_{t}\right\}^T_{t=0}} \sum_{t=0}^{T}\mathsf{E}\left[\textit{l}_{d}\left(X_t,\hat{X}_t\right)\right]+\lambda I \left(\hat{X}^{T};Y^{T}\right),
\end{equation}
where $\textit{l}_{d}\left(X_t,\hat{X}_t\right)$ denotes the estimation loss, $I\paren{\hat{X}^{T}; Y^{T}}$ is the leakage level of private information through the estimator, and $\lambda$ is a positive constant.

\section{Structural Properties of the Optimal Privacy-aware Estimator}\label{Sec:Stru}
In this section, first derive the Bellman optimality equation for the optimal privacy-aware estimation problem. We then study the structural properties of the optimal estimator.
\subsection{Optimality Equations}
In this subsection, we first define an auxiliary optimization problem, and  show it is equivalent to the optimization problem \eqref{Eq.OP}. Then, we derive optimality  equations for the equivalent optimization problem. Let $\mathcal{A}_t=\left\{{a_t\paren{\left. \hat{x}_t\right|\tilde{z}^{t}}}, \forall \hat{x}_t, \tilde{z}^t\right\}$ where $a_t\paren{\left. \hat{x}_t\right|\tilde{z}^{t}}$ is a conditional distribution of $\hat{X}_t$ given $\tilde{z}^t$. Consider an auxiliary decision-making problem where, given $\paren{\hat{X}^{t-1},\tilde{Z}^t}$, we first select $\mathcal{A}_t$ based on $\hat{X}^{t-1}$, and then generate $\hat{X}_t$ according to $a_t\paren{\hat{x}_t \left| \tilde{Z}^{t} \right.}\in \mathcal{A}_t$. The optimal choice of $\mathcal{A}_t$ is the solution of the following auxiliary optimization problem
\begin{equation}\label{Eq.OP2}
 	\min_{\left\{\mathcal{A}_t\right\}^T_{t=0}} L\left(\mathcal{A}\right)=\min_{\left\{\mathcal{A}_t\right\}^T_{t=0}} \sum_{t=0}^{T}\mathsf{E} \left[\textit{l}_{d}\left(X_t,\hat{X}_t\right)\right]+\lambda I \left(\hat{X}^{T};Y^{T}\right),
\end{equation}
where we implicitly allow  $\mathcal{A}_{t}$ to  depend on ${\hat{X}^{t-1}}$.

The following lemma shows that the auxiliary optimization problem \eqref{Eq.OP2}  and the original optimization problem are equivalent. 
\begin{lemma} \label{Lm.Equivalent}
	The original optimization problem \eqref{Eq.OP} is equivalent to the auxiliary optimization problem \eqref{Eq.OP2}.
\end{lemma}
\begin{proof}
	See Appendix \ref{App:Lm.Equivalent}.
\end{proof}

The next theorem presents the optimality equations for the auxiliary decision-making problem. 

\begin{theorem} \label{Th.OPTEQU}
The Bellman optimality equation for the auxiliary optimization problem \eqref{Eq.OP2} is given by 
\begin{align} \label{Eq.VF}
V_{t}^{\star}\left( b_t \right) &=\min_{\mathcal{A} _t} \sum_{y^{t-1},\tilde{z}^t,\hat{x}_t}a_t\left( \left.\hat{x}_t \right|\tilde{z}^t \right) b_t\left( y^{t-1},\tilde{z}^t \right) \Big[ \int{p\left( \left. x_t\right|y^{t-1},\tilde{z}^t \right) \textit{l}_{d}\left( x_t,\hat{x}_t \right)dx_t}
\\&+\lambda \log \frac{\sum_{\tilde{z}^t}{a_t\left( \left.\hat{x}_t\right|\tilde{z}^t \right) b_t\left( y^{t-1},\tilde{z}^t \right)}}{\left( \sum_{y^{t-1},\tilde{z}^t}{a_t\left( \left.\hat{x}_t\right|\tilde{z}^t \right) b_t\left( y^{t-1},\tilde{z}^t \right)} \right) \left( \sum_{\tilde{z}^t}{b_t\left( y^{t-1},\tilde{z}^t \right)} \right)} \Big]+\mathsf{E}\left[ V_{t+1}^{\star}\left( b_{t+1}  \right) \left|b_t\right. \right] ,
\end{align} 
	where $b_t$ is the belief state defined as 	$b_t\left( y^{t-1},\tilde{z}^t\right)=P\left( y^{t-1},\tilde{z}^t\left|\hat{X}^{t-1} \right. \right)$ 
	 and $V_{t}^{\star}\left( b_t \right)$ is the optimal value function associated with $b_t$ with $ V_{T+1}^{\star}\left(\cdot\right)=0$. Moreover, the belief state $b_{t+1}$ can be recursively computed using $b_t$ and the following forward equation 	
	\begin{align}\label{Eq.BSUP}
	b_{t+1}\left( y^t,\tilde{z}^{t+1} \right) &=\frac{a _t\left(\left. \hat{X}_t\right|\tilde{z}^t\right) b_t\left( y^{t-1},\tilde{z}^t \right) P\left( \left. y_t\right|y_{t-1} \right) \int{P\left( \left. \tilde{z}_{t+1}\right|x_{t+1} \right) p\left( \left. x_{t+1}\right|y^t,\tilde{z}^t \right) dx_{t+1}}}{\sum_{y^{t-1},\tilde{z}^t}{a _t\left(\left. \hat{X}_t\right|\tilde{z}^t\right) b_t\left( y^{t-1},\tilde{z}^t \right)}}
	\end{align}  
with $b_0=P_{\tilde{z}_0}$ where $P_{\tilde{z}_0}$ is the probability mass function of $\tilde{Z}_0$.
\end{theorem}
\begin{proof}
	See Appendix \ref{App:Th.OPTEQU}.
\end{proof}
According to the Theorem \ref{Th.OPTEQU}, the Bellman optimality equation for the auxiliary decision-making problem depends on the belief state $b_t$. Thus, the optimal solution of  \eqref{Eq.VF} is a function of $b_t$, \emph{i.e.,} $\mathcal{A}^{\star}_t\left(b_t\right)$. Therefore, to compute the optimal policy of \eqref{Eq.OP2}, one needs to first compute $b_t$, and then select $\mathcal{A}^{\star}_t\left(b_t\right)$ according to the Bellman optimality equation in \eqref{Eq.VF}.

 The next theorem studies the structure of the optimal policy of the original estimation problem  \eqref{Eq.OP}.
\begin{theorem}\label{Theo: EstPol}
	Let $\mathcal{A}^{\star}_t\left(b_t\right)$ denote the solution \eqref{Eq.VF}. Then, given $\left(\hat{X}^{t-1},\tilde{Z}^t\right)$, the optimal estimation policy at time $t$ is $\pi_t^{\star}\left(\hat{x}_t\left|\tilde{Z}^t,b_t\right.\right)=a_t^{\star}\left(\hat{x}_t\left|\tilde{Z}^t\right.\right)\in\mathcal{A}^{\star}_t\paren{b_t}$, where $b_t$ is the belief state associated with $\hat{X}^{t-1}$. 
\end{theorem}
\begin{proof}
The proof follows directly from Lemma \ref{Lm.Equivalent} and Theorem \ref{Th.OPTEQU}.
\end{proof}
According to Theorem \ref{Theo: EstPol}, the optimal estimation policy is obtained in two steps, where we first use the history of the estimator outputs to build the belief state $b_t$, and then we optimize the policy collection in \eqref{Eq.VF} to find the optimal estimation policy. 

\subsection{Structural Properties of the Optimal Estimation Policy}

The optimal privacy-aware estimation is a closed-loop control problem wherein the estimator controls the adversary's belief about the private process. To highlight this point, note that the belief state $b_t$ is recursively updated using the estimator output and the estimation policy according to the update law $b_{t+1}=\varPhi \left( b_t,\mathcal{A}_t^\star,\hat{X}_t \right)$, as shown in Fig. \ref{Fig.EstStruA}. Thus, the estimator acts as a feedback controller which controls the evolution of the belief state. Also, note that an adversary with access to the outputs of the estimator will form the posterior probability $P\left(y^{t}\left. \right| \hat{X}^t\right)$ to infer the private variable, which can be expanded as 
\begin{equation} \label{Eq.ADVINF}
	P\left(y^{t}\left. \right| \hat{X}^t\right) = \sum_{\tilde{z}^{t+1}} b_{t+1}\paren{y^{t},\tilde{z}^{t+1}}.\nonumber
\end{equation}
Hence, the optimal privacy-aware estimator is effectively a controller that controls the adversary's belief about the private process.

At this point, it is constructive to compare the optimal privacy-aware estimation with the classical minimum mean square error (MMSE) estimator of the states of the system based on the measurements. The MMSE estimator is computed using  $\hat{X_t}=\mathsf{E}\left[X_t\left|\tilde{Z}^t\right.\right]$, as shown in Fig. \ref{Fig.EstStruB}. Thus, different from the optimal privacy-aware estimator, the classical MMSE estimation is not a closed-loop control problem.

\begin{figure}[!htbp]
	\centering
	\begin{subfigure}[b]{.45\textwidth}
        \hspace{0.8cm}
		\includegraphics[scale=0.2]{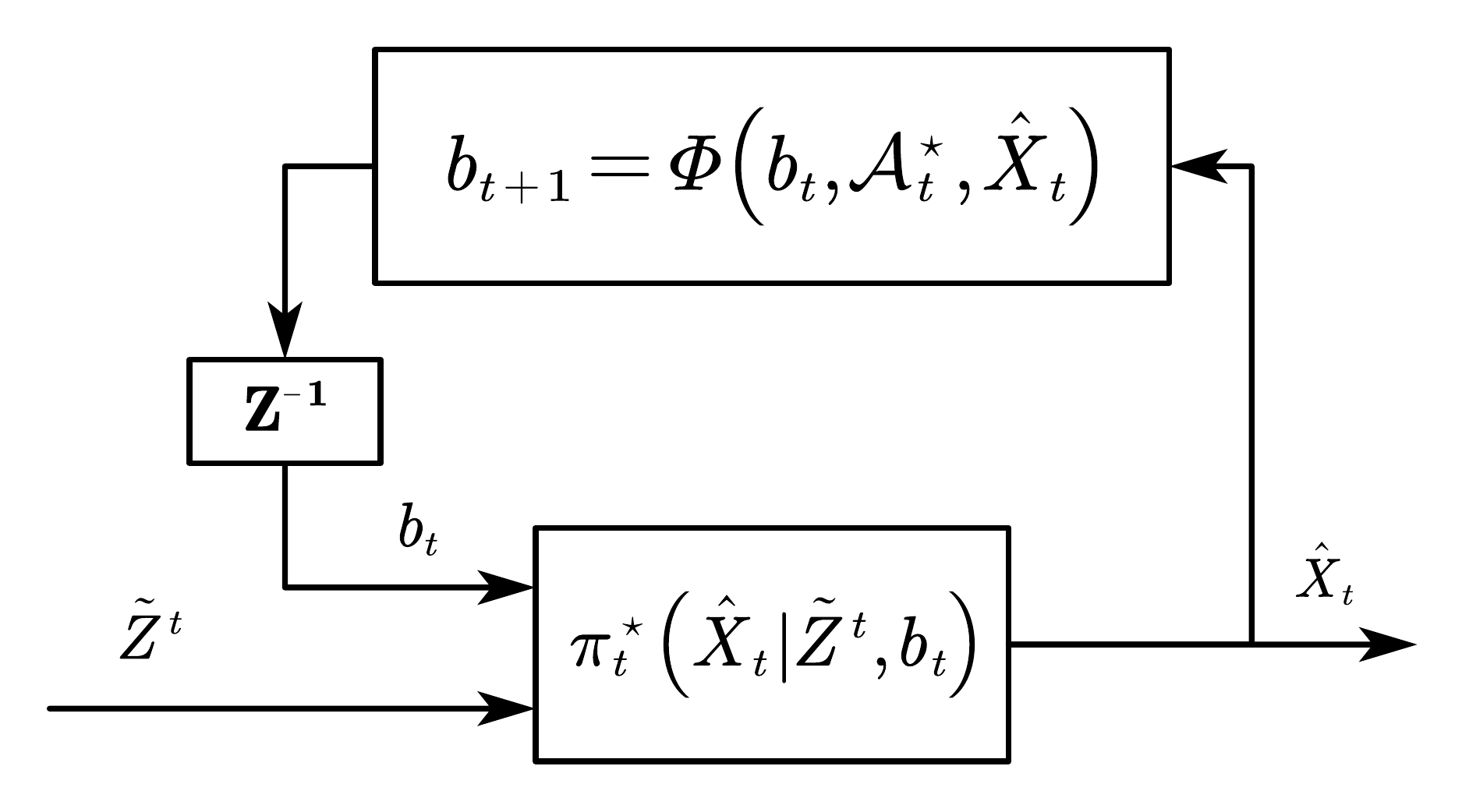}
            \caption{}
		\label{Fig.EstStruA}
        \end{subfigure}
        \begin{subfigure}[b]{.45\textwidth}
        \hspace{1.6cm}
		\includegraphics[scale=0.2]{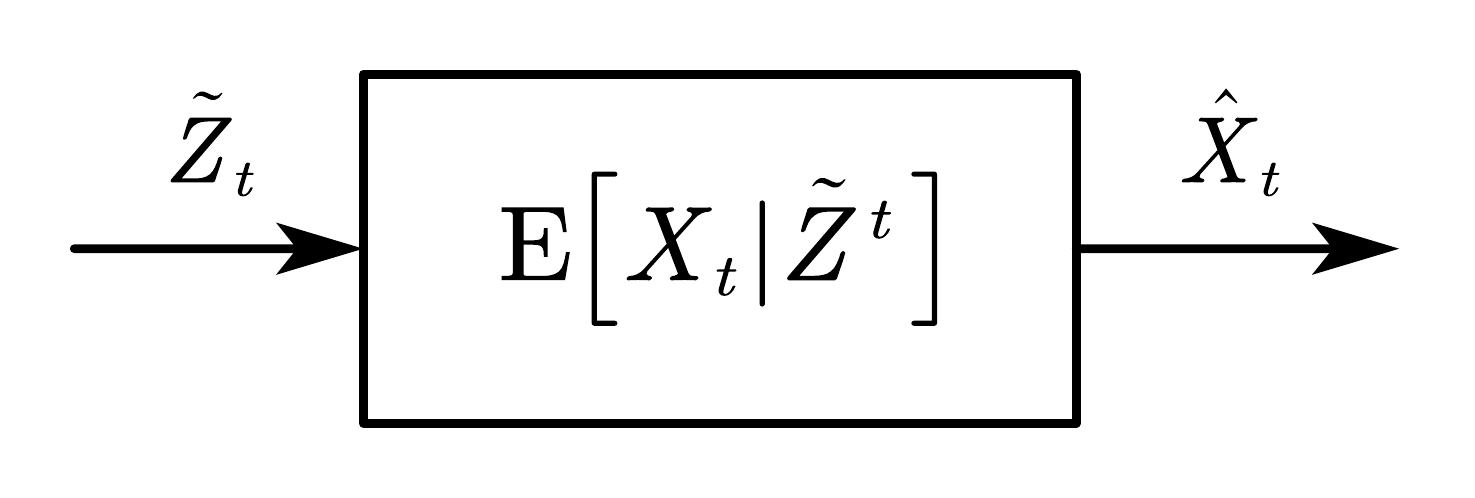}
            \caption{}
		\label{Fig.EstStruB}
        \end{subfigure}
	\caption{The structure of the optimal privacy-aware estimator $(a)$, and the MMSE estimator $(b)$.
	}
\end{figure}

Finally, we note that the dynamic programming solution in \eqref{Eq.OP2} decomposes the optimization problem in \eqref{Eq.OP} into sub-problems only according to the belief state $b_t$, rather than $b_t$ and $\tilde{Z}^t$. Thus, the optimal value function only depends on $b_t$. Although the sensor measurements are not used for the dynamic programming decomposition, the optimal estimation policy depends on both sensor measurements and the belief state. That is, we first choose an optimal policy collection $\mathcal{A}^\star_t=\left\{{a_t\paren{\left. \hat{x}_t\right|\tilde{z}^{t}}}, \forall \hat{x}_t, \tilde{z}^t\right\}$ based on $b_t$, and then choose a distribution from $\mathcal{A}^\star_t$ based on the history of sensor measurements.

\section{A Computational Algorithm For Estimator Design}\label{Sec:Comp}
The computation of the optimal estimation policy using Theorem \ref{Th.OPTEQU} is difficult when the time-horizon is large due to two main reasons. First, the dimension of the belief state increases with time which enlarges the search space of the optimal policy. Second, the computation of the logarithmic term in \eqref{Eq.VF} becomes prohibitive when the horizon becomes large as it involves summation over all possible combinations of the past and current quantized measurements. 

In this section, we develop a numerical algorithm for computing an optimal privacy-aware estimator. To this end, let $\pi_\theta \paren{\cdot|h_t}$ denote a stationary stochastic policy, parameterized by $\theta$, where $\pi_\theta \paren{\cdot|h_t}$ is a probability distribution on  $\left\{{x}_{(1)},\dots,{x}_{(M)}\right\}$, and $h_t$ is the available information to the estimator. Here, at each time-step $t$, $\hat{X}_t$ is generated according to the distribution $\pi_\theta \paren{\cdot|h_t}$. This formulation allows us to model two important special cases: $(i)$ a finite memory estimator where the output of the estimator depends on the current measurements and the past $d$ measurements and the past $d$ outputs of the estimator, \emph{i.e.,} $h_t=\left\{\tilde{Z}_{t-d}^{t},\hat{X}_{t-d}^{t-1}\right\}$, $(ii)$ a state-based estimator where the output of the estimator dependents on a state which summarizes the past measurements and the past outputs of the estimator, \emph{i.e.,} $h_t=M_t$ where $M_t=f\left(M_{t-1},\tilde{Z}_t,\hat{X}_{t-1}\right)$.

The optimal estimation policy is the solution of the following optimization problem  
	\begin{equation}\label{Eq:Opt-PG}
	\min_{\pi_\theta \paren{\cdot|\cdot}} \sum_{t=0}^T\mathsf{E}\left[\textit{l}_{d}\left(X_t,\hat{X}_t\right)\right]+\lambda I_\theta \left(\hat{X}^{T};Y^{T}\right).
	\end{equation}
	
The next theorem derives an expression for the gradient of the objective function function in \eqref{Eq:Opt-PG} with respect to the policy parameter $\theta$.
\begin{theorem} \label{Th.PG}
	Let $L\paren{\theta}$ denote the objective function in \eqref{Eq:Opt-PG} when the policy parameter is $\theta$. The gradient of $L\paren{\theta}$ with respect to $\theta$ can be expressed as 
	\begin{equation}\label{Eq: Grad}
		\nabla _{\theta}L\paren{\theta} =\mathsf{E}\left[\paren{\sum_{t=0}^T{l_t}\paren{X_t,Y^{t},\hat{X}_t;\theta}} \paren{\sum_{t=0}^{T}{\nabla _{\theta} \log \pi_\theta\paren{\left. \hat{X}^{t} \right| h_{t}}}} \right],
	\end{equation}
	where,
	\begin{equation} \label{Eq.Loss}
		{l_t}\left( X_t,Y^{t},\hat{X}^t;\theta \right) = \textit{l}_{d}\left( X_t,\hat{X}_t \right)+ \log{\frac{P_{\theta}\left( \left. \hat{X}_t\right|\hat{X}^{t-1},Y^{t} \right)}{P_{\theta}\left( \left. \hat{X}_t\right|\hat{X}^{t-1} \right)}}.
	\end{equation}
\end{theorem}
\begin{proof}
	See Appendix \ref{App: Th.PG}.
\end{proof}
Theorem \ref{Th.PG} allows us to numerically compute the solution of the optimization problem \eqref{Eq:Opt-PG} using the gradient descent algorithm where the expectation in \eqref{Eq: Grad} is numerically computed for a given value of $\theta$ and then $\theta$ is updated using the gradient value. Alternatively, one can develop stochastic gradient descent algorithms for solving \eqref{Eq:Opt-PG} based on  \eqref{Eq: Grad}. 
\begin{remark}
	We note that the standard policy gradient theorem from the reinforcement learning literature, e.g., see \cite{bertsekas2019reinforcement}, is not directly applicable to the optimization \eqref{Eq:Opt-PG} as the objective function in \eqref{Eq:Opt-PG} is nonlinear in the policy. Thus, in Theorem \ref{Th.PG}, we develop an expression for the gradient of the objective function with respect to the policy parameter.	
\end{remark}
\begin{remark}
    The optimal policy can be computed using the actor-critic methods or the temporal difference learning method, \emph{e.g.,} see \cite{bertsekas2019reinforcement}.
\end{remark}

The term ${l_t}\left( X_t,Y^{t},\hat{X}^t;\theta \right)$ in \eqref{Eq.Loss} consists of two parts: $(i)$ $\textit{l}_{d}\left( X_t,\hat{X}_t \right)$ which captures the estimation loss, and $(ii)$ $\log{\frac{P_{\theta}\left( \left. \hat{X}_t\right|\hat{X}^{t-1},Y^{t} \right)}{P_{\theta}\left( \left. \hat{X}_t\right|\hat{X}^{t-1} \right)}}$ which captures the leakage of private information. We refer to the second term on the right-hand side of \eqref{Eq.Loss} as the information loss. The direct computation of the information loss is challenging as it involves the summations over all possible values of current and past quantized sensor measurements. To highlight this fact, note that the information loss can be written as 
\begin{equation} \label{Eq.DirComp}
	\log{\frac{P_{\theta}\left( \left. \hat{X}_t\right|\hat{X}^{t-1},Y^{t} \right)}{P_{\theta}\left( \left. \hat{X}_t\right|\hat{X}^{t-1} \right)}} = \log{\frac{\sum_{\tilde{z}^t}{\pi_\theta\left( \left.\hat{X}_t\right|\tilde{z}^t, \hat{X}^{t-1} \right) b_t\left( Y^{t-1},\tilde{z}^t \right)}}{\left( \sum_{y^{t-1},\tilde{z}^t}{\pi_\theta\left( \left.\hat{X}_t\right|\tilde{z}^t, \hat{X}^{t-1} \right) b_t\left( y^{t-1},\tilde{z}^t \right)} \right) \left( \sum_{\tilde{z}^t}{b_t\left( Y^{t-1},\tilde{z}^t \right)} \right)} }.
\end{equation}
Both the numerator and denominator in \eqref{Eq.DirComp} involve summations over all possible values of $\tilde{z}^t$ which implies that the complexity of direct computation of the information loss increases exponentially with time.

The next theorem develops a variational representation for the information loss. Under the variational approach, a divergence or a likelihood function is expressed based on the solution of a functional optimization problem \cite{nguyen2010estimating, poole2018variational}.
\begin{theorem} \label{Th.ILA}
	The information loss can be written as,
	\begin{equation} \label{Eq.ILA}
		\log \frac{P_\theta\left( \hat{x}_t|\hat{x}^{t-1},y^t \right)}{P_\theta\left( \hat{x}_t|\hat{x}^{t-1} \right)}=g_t^\star\left( \hat{x}_t, \hat{x}^{t-1},y^t \right) -f_t^\star\left( \hat{x}_t, \hat{x}^{t-1} \right), \quad t=0,\dots,T       
	\end{equation}
	where $\left\{g_t^\star\paren{\cdot}\right\}_{t=0}^T$ and $\left\{f_t^\star\paren{\cdot}\right\}_{t=0}^T$ are optimal solutions of the following two optimization problems,
	\begin{equation} \label{Eq.ILAG}
		\mathop {\mathrm{sup}} \limits_{\{g_t\paren{\cdot}\}_{t=0}^T}\mathsf{E}\left[ \sum_{t=0}^T{g_{t}\left( \hat{X}_t, \hat{X}^{t-1},Y^t \right) - e^{g_{t}\left( \tilde{X},\hat{X}^{t-1},Y^t \right)-1}} \right],
	\end{equation}
	\begin{equation} \label{Eq.ILAF}
		\mathop {\mathrm{sup}} \limits_{\{f_t\paren{\cdot}\}_{t=0}^T}\mathsf{E}\left[ \sum_{t=0}^T{f_{t}\left( \hat{X}_t, \hat{X}^{t-1} \right) - e^{f_{t}\left( \tilde{X},\hat{X}^{t-1} \right)-1} } \right] ,
	\end{equation}
	where $g_t\paren{\cdot}$ and $f_t\paren{\cdot}$ belong to the space of measurable functions, $\tilde{X}$ is a discrete uniform variable which has the same support set as $\hat{X}_t$ and is independent of $\hat{X}^T$ and $Y^{T}$.
\end{theorem}
\begin{proof}
	See Appendix \ref{App: Th.PG}.
\end{proof}
According to Theorem \ref{Th.ILA}, the information loss can be computed using the functions $g_t\paren{\cdot}$ and $f_t\paren{\cdot}$. However, solving the optimization problems can be difficult as $(i)$ they are functional optimization problems, and $(ii)$ the exact computation of the expectations in   \eqref{Eq.ILAG} and \eqref{Eq.ILAF} becomes prohibitive when $T$ is large. To tackle the first issue, the functions $g\paren{\cdot}=\left[g_0\paren{\cdot},\dots,g_T\paren{\cdot}\right]^\top$ and $f\paren{\cdot}=\left[f_0\paren{\cdot},\dots,f_T\paren{\cdot}\right]^\top$ can be restricted to certain  parameterized families of functions $g_{\phi}\paren{\cdot}$ and $f_{\psi}\paren{\cdot}$, respectively, where $\phi,\psi$ are two real-valued parameters, \emph{e.g.,} see \cite{tsur2023neural} for examples of such parametrization. Thus, the search space of the optimization problems in the variational representation can be restricted to finite-dimensional Euclidean spaces. To tackle the second issue, stochastic gradient ascent algorithms can be used to numerically solve optimization problems, wherein the exact computation of the expectation is replaced with the mean of sample trajectories of $\hat{X}^T$ and $Y^{T}$.

Algorithm \ref{Alg: PG} presents a numerical solution for computing the optimal policy using Theorems \ref{Th.PG} and \ref{Th.ILA}. Here, for a given $\theta$, we first numerically find the optimal values of $\phi$ and $\psi$ with a stochastic gradient ascent algorithm that uses the mean of sample trajectories of $\hat{X}^T$ and $Y^{T}$ to approximate the expectation operators as shown in \eqref{Eq: Grad-phi} and \eqref{Eq: Grad-psi}. The expectation approximation uses the subscript $\omega$ as the index of sample trajectories. Given the optimal values of  $\phi$ and $\psi$, the information loss can be approximated as $\left(g_{\phi,t}-f_{\psi,t}\right)$. Finally, we update the value of $\theta$ using the information loss approximator and a stochastic gradient descent algorithm which approximates the expectation using the sample trajectories of $X^T$,$\tilde{Z}^T$, $\hat{X}^T$ and $Y^{T}$. This procedure is repeated until the algorithm converges. 

Note that the complexity of computing the information loss in Algorithm \ref{Alg: PG} increases linearly with time whereas the complexity of the direct computation of the information loss grows exponentially with time. Thus, the variational approach provides a low-complexity solution for computing information loss. 

\begin{algorithm}[htbp] 
	\caption{Policy Gradient with Information Loss Approximator} \label{Alg: PG}
	Initialize the policy $\pi_\theta$, the approximator $\left(g_{\phi,t}-f_{\psi,t}\right)$ and step sizes $\alpha$, $\beta$, $\gamma$.
	\begin{algorithmic}[1] 
		\Repeat
        \State Given $\theta$, generate $\left\{Y_{t,\omega},X_{t,\omega},\tilde{Z}_{t,\omega},\hat{X}_{t,\omega}\right\}_{t=0}^T, \omega=1,\dots,K.$
		\State {\bf Information loss approximator step:}
		\Repeat
		\begin{equation}\label{Eq: Grad-phi}
		    \qquad \nabla _{\phi}F_1\left( \phi \right) =\nabla _{\phi}\left[ \frac{1}{K}\sum_{\omega=1}^K{\sum_{t=0}^T{g_{\phi,t}\left( \hat{X}_{t,\omega},\hat{X}_{\omega}^{t-1},Y_{\omega}^{t} \right) -e^{g_{\phi,t}\left( \tilde{X}_{\omega},\hat{X}_{\omega}^{t-1},Y_{\omega}^{t} \right) -1}}} \right] , 
		\end{equation}
		\begin{equation}\label{Eq: Grad-psi}
		    \qquad \nabla _{\psi}F_2\left( \psi \right) =\nabla _{\psi}\left[ \frac{1}{K}\sum_{k=1}^K{\sum_{t=0}^T{f_{\psi,t}\left( \hat{X}_{t,\omega},\hat{X}_{\omega}^{t-1} \right) -e^{f_{\psi,t}\left( \tilde{X}_{\omega},\hat{X}_{\omega}^{t-1} \right) -1}}} \right] , 
		\end{equation}
		\State update the approximator via $\phi \gets \phi+\alpha\nabla _{\phi}F_1\left( \phi \right)$ and $\psi \gets \psi+\beta\nabla _{\psi}F_2\left( \psi \right)$.
		\Until{Convergence}
		\State {\bf Policy gradient step:} 
		\begin{scriptsize}
			\begin{equation}
				\qquad \nabla _{\theta}L\left( \theta \right) = \nabla _{\theta}\left[ \frac{1}{K}\sum_{k=1}^K{\left( \sum_{t=0}^T{l_d\left( X_{t,\omega},\hat{X}_{t,\omega} \right) +\lambda g_{\phi,t}\left( \hat{X}_{t,\omega},\hat{X}_{\omega}^{t-1},Y_{\omega}^{t} \right) -\lambda f_{\psi,t}\left( \hat{X}_{t,\omega},\hat{X}_{\omega}^{t-1} \right)} \right) \left( \sum_{t=0}^T{\nabla _{\theta}\log \pi _{\theta}\left( \left. \hat{X}_{t,\omega} \right|h_{t,\omega} \right)} \right)} \right] , \nonumber
			\end{equation}
		\end{scriptsize}
		\State \qquad update the policy via $\theta \gets \theta-\gamma\nabla _{\theta}L\left( \theta \right)$
		\Until{Convergence}
	\end{algorithmic}
\end{algorithm}

\section{Numerical Results}\label{Sec:Sim}
In this section, we numerically study the privacy-aware estimation of the CO$_2$ level of a building zone based on noisy CO$_2$ measurements. To this end, we assume CO$_2$  evolves according to  $X_{t+1} = 0.75X_t + 0.2Y_t + W_t,$ where $\left\{W_t\right\}_t$  is the sequence of iid zero mean Gaussian random variables with standard deviation $0.05$, and $\left\{Y_t\right\}$ is the  zone occupancy  modeled as a Markov chain with the support set $\left\{0,1,2\right\}$, and the following transition probability matrix 
	$$\left[ \begin{matrix}	0.8&		0.1&		0.1\\	0.1&		0.8&		0.1\\	0.0&		0.08&		0.92\\\end{matrix} \right] . $$

The CO$_2$  measurement at time $t$ is given by $Z_t = X_t + V_t$, where $\left\{V_t\right\}_t$ is a sequence of iid Gaussian random variables with zero means and standard deviation  $0.1$. The sensor measurements are quantized using a uniform quantizer with quantization sensitivity $0.1$. The quantized measurements are then used to estimate the CO$_2$ using the optimal privacy-aware estimator.  The horizon length was $64$ in our simulations.

We first study the trade-off between the estimation error and privacy in two cases: ($i$) optimal privacy-aware estimation ($ii$) the additive noise approach, shown in Fig. \ref{Fig.AddNoise}. In the additive noise approach, we first use a classical minimum mean square estimator to estimate the state. Then, the state estimates are perturbed using Gaussian noise with zero mean and variance $\sigma^2$, i.e., $\mathcal{N}\left(0,\sigma^2\right)$, to ensure privacy.  In each case,  we infer the occupancy based on the state estimates via maximum likelihood estimation method.

\begin{figure}[H]
	\centering
	\includegraphics[width=0.3\textwidth]{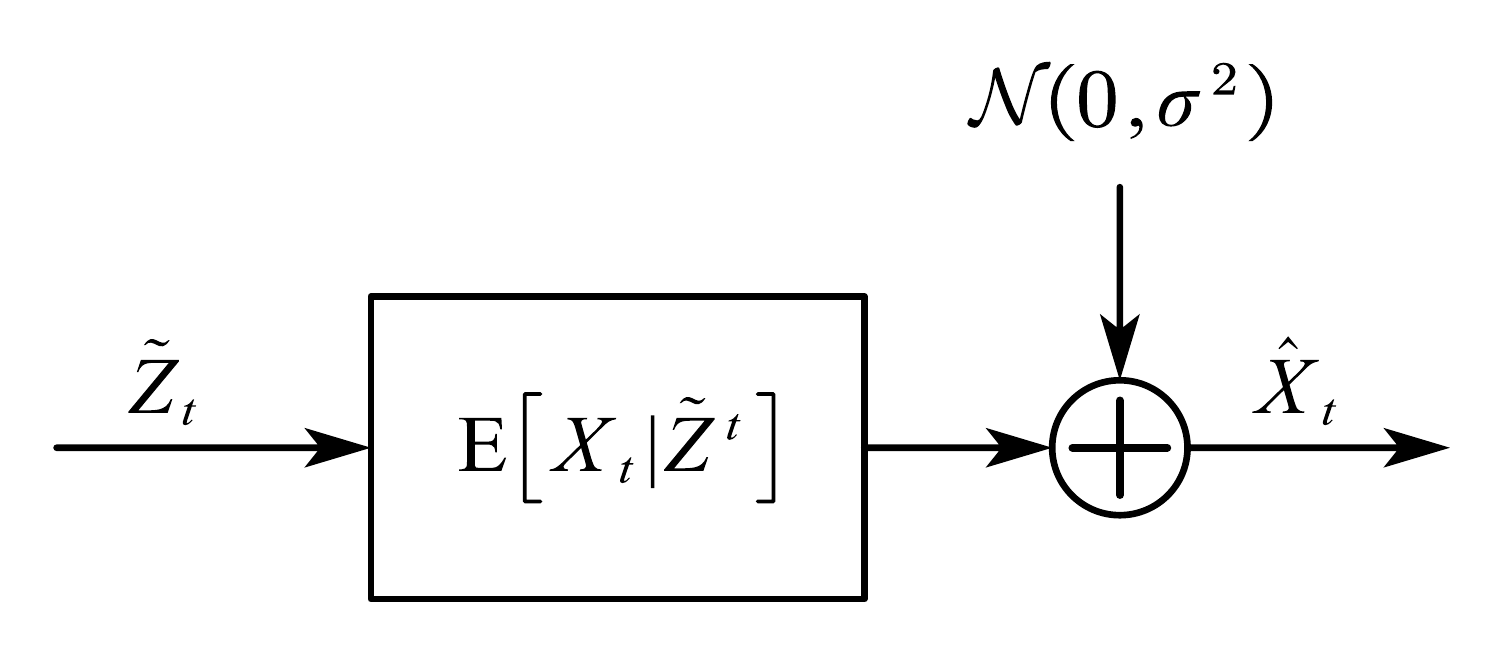}
	\caption{Additive noise approach.}
	\label{Fig.AddNoise}
\end{figure}
Fig. \ref{Fig.PriDistCmp} shows the total state estimation error under the optimal privacy-aware estimation and the additive noise approaches versus the accuracy of the maximum likelihood occupancy estimator.  A small value of accuracy implies that the occupancy estimator cannot infer the occupancy correctly. Thus, smaller accuracy values correspond to higher privacy levels.
\begin{figure}[H]
	\centering
	\includegraphics[width=0.4\textwidth]{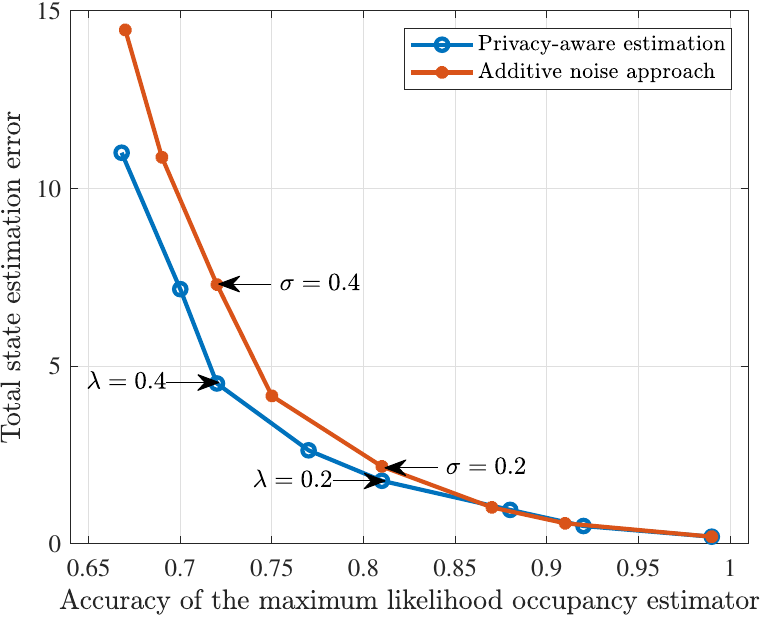}
	\caption{Total estimation error under the optimal privacy-aware estimation and the additive noise approaches versus the accuracy of the occupancy estimator.}
	\label{Fig.PriDistCmp}
\end{figure}
According to Fig. \ref{Fig.PriDistCmp}, the total state estimation error increases as the accuracy of the occupancy estimator becomes small. This is because, to increase the privacy level, we need to increase the penalty on the leakage level of private information $(\lambda)$ in the optimization \eqref{Eq.OP}. Also, based on Fig. \ref{Fig.PriDistCmp}, the optimal privacy-aware estimator, results in a smaller estimation error compared with the additive noise approach. 

 Fig. \ref{Fig.EstResA}-\ref{Fig.EstResC} show trajectories of the output of the optimal privacy-aware estimator for different values of $\lambda$. According to this figure, the estimation error is small when $\lambda$ equals zero, corresponding to the smallest privacy level. However, as $\lambda$ becomes large, the privacy level increases at the cost of increasing the total state estimation error. Finally, Fig. \ref{Fig.EstResD} shows the misdetection instance of the maximum likelihood occupancy estimator when $\lambda$ is equal to $0.2$ and $0.4$. In this figure, the indicator function $\I{\hat{Y}_t=Y_t}$ is equal to one when the occupancy estimator makes a mistake and is zero otherwise. According to this figure, the number of misdetection instances increases as $\lambda$ becomes large, which indicates that the occupancy estimator cannot accurately infer the occupancy as $\lambda$ increases.


\begin{figure}[!htbp]
	\centering
	\begin{subfigure}[b]{.45\textwidth}
		\includegraphics[scale=0.6]{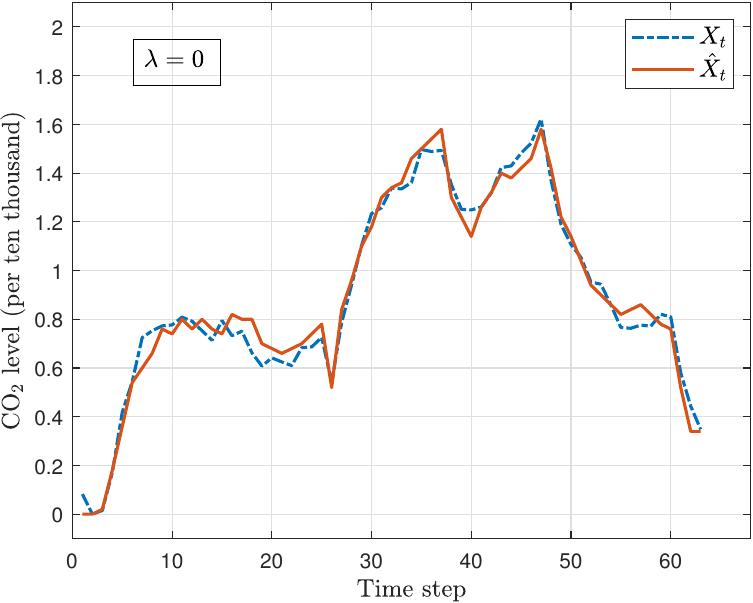}
            \caption{}
		\label{Fig.EstResA}
        \end{subfigure}
        \begin{subfigure}[b]{.45\textwidth}
		\includegraphics[scale=0.6]{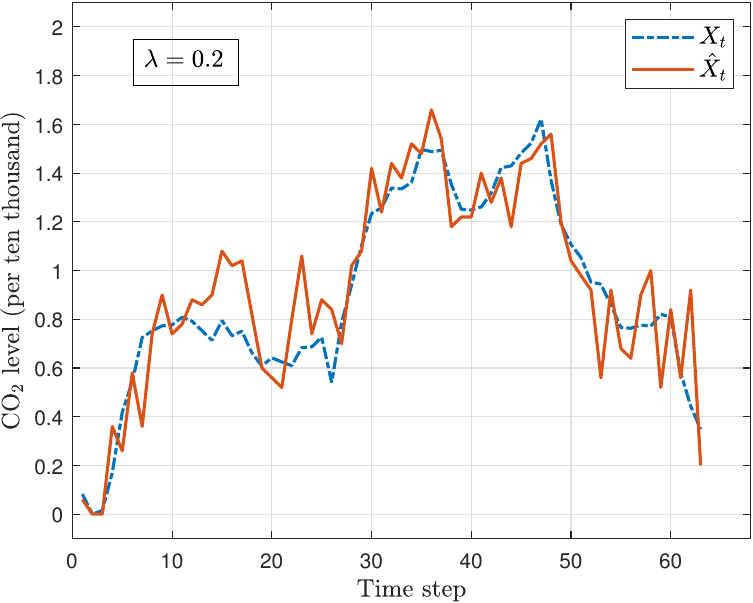}
            \caption{}
		\label{Fig.EstResB}
        \end{subfigure}
        \begin{subfigure}[b]{.45\textwidth}
		\includegraphics[scale=0.6]{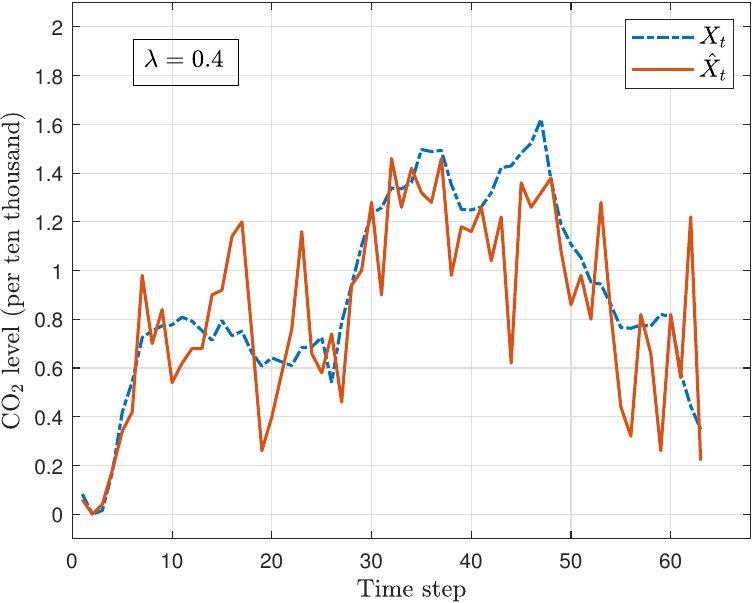}
            \caption{}
		\label{Fig.EstResC}
        \end{subfigure}
        \begin{subfigure}[b]{.45\textwidth}
        \quad
		\includegraphics[scale=0.6]{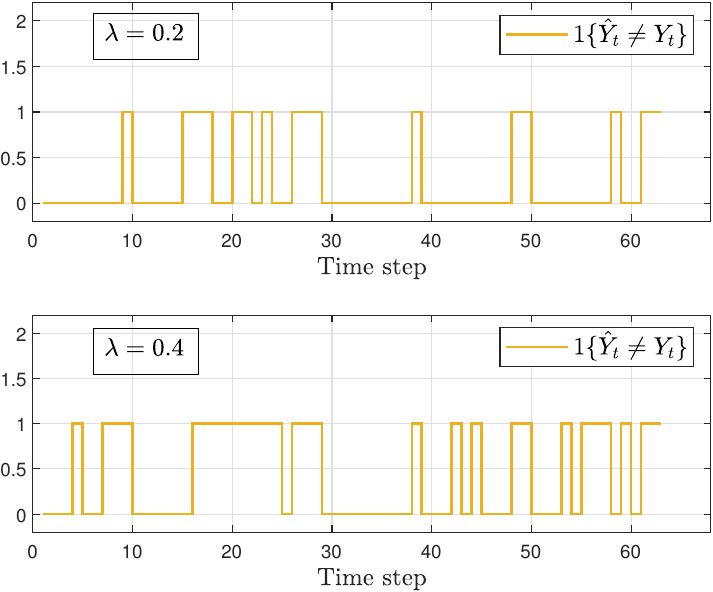}
            \caption{}
		\label{Fig.EstResD}
        \end{subfigure}

	\caption{Trajectories of the estimator output for different values of $\lambda$ ($a$)-($c$),  Misdetection instances of the maximum likelihood occupancy estimator for $\lambda=0.2$ and $\lambda=0.4$ ($d$).
	}
	\label{Fig.EstRes}
\end{figure}


\section{Conclusion}\label{Sec:Con}
In this paper, we developed an information-theoretic framework for the optimal privacy-aware estimation in dynamic settings. We studied the structural properties of the optimal estimator via dynamic programming decomposition. We also proposed a policy gradient algorithm for computing the optimal estimator based on a variational formulation of the mutual information. Our numerical results show that the information-theoretic approach to privacy outperforms the additive noise approach.

\bibliographystyle{ieeetr}
\bibliography{reference}

\clearpage
\appendices	
	\section{Proof of Lemma \ref{Lm.Equivalent}}\label{App:Lm.Equivalent}
	We first show that given an estimation policy $\pi=\left\{\pi_t\right\}^T_{t=0}$ we can construct a policy $\mathcal{A}=\left\{\mathcal{A}_t\right\}^T_{t=0}$ for the auxiliary decision-making problem which achieves the same value of objective function as $\pi$. To this end, given $\pi$, we choose the policy collection as
	\begin{equation}
	\mathcal{A}_t=\left\{ a_t\left( \left. \hat{x}_t\right|\tilde{z}^t \right)=\pi_{t}\left({\hat{x_t}}\left|{\tilde{z}^t,\hat{X}^{t-1}}\right. \right), \forall \hat{x}_t, \tilde{z}^t \right\}. \nonumber
	\end{equation}
	Clearly, we have $L\left(\mathcal{A}\right)=L\left(\pi\right)$ where  $L\left(\mathcal{A}\right)$ and $L\left(\pi\right)$ are the values of objective function under $\pi$ and $\mathcal{A}$, respectively. We next show that given $\mathcal{A}$ we can construct an estimation policy $\pi$ which achieves the same value of objective function as $\mathcal{A}$. Given $\mathcal{A}$, let $	\pi_{0}\left(\left. {\hat{X_0}}\right|{\tilde{Z}_0}\right)=a_0\left(\left. \hat{X}_0\right|\tilde{Z}_0\right)$. For $t>0$, given the policy collection $\mathcal{A}_t=\mathcal{P}_t\paren{\hat{X}^{t-1}}$ via \eqref{Eq.OP2}, we select the estimation policy by,
	\begin{equation}
	\pi_{t}\left(\left. {\hat{X_t}}\right|{\tilde{Z}^t,\hat{X}^{t-1}}\right)=a_t\left(\left. \hat{X}_t\right|\tilde{Z}^t\right), \nonumber
	\end{equation}
Clearly, we have $L\left(\pi\right)=L\left(\mathcal{A}\right)$. Thus, the auxiliary decision-making problem and the original optimization problem are equivalent. 
	
\section{Proof of Theorem \ref{Th.OPTEQU}}\label{App:Th.OPTEQU}
To prove this result, we first derive an expansion for the mutual information term in \eqref{Eq.OP2}.	
	\begin{lemma} \label{Lm.MISimp}
		The mutual information term in \eqref{Eq.OP2} can be expanded as 
		\begin{align}
				I\left( \hat{X}^T;Y^T \right) &=\sum_{t=1}^T{I\left(\left. \hat{X}_t;Y^{t-1}\right|\hat{X}^{t-1} \right)}.   \nonumber
		\end{align}
	\end{lemma}
\begin{proof}
	See Appendix \ref{App: MIE}.
\end{proof}
	Using Lemma \ref{Lm.MISimp} and the definition of conditional mutual information, the objective function of \eqref{Eq.OP2} can be written as 
\begin{align}
\sum_{t=0}^T{\mathsf{E}\left[ \textit{l}_{d}\left( X_t,\hat{X}_t \right) \right]}+\lambda I\left( \hat{X}^T;Y^T \right) &=\sum_{t=0}^T{\mathsf{E}\left[ \textit{l}_{d}\left( X_t,\hat{X}_t \right)  \right] +\lambda I\left(\left. \hat{X}_t;Y^{t-1}\right|\hat{X}^{t-1} \right)} \nonumber\\
 &=\sum_{t=0}^T\mathsf{E}\left[  \textit{l}_{d}\left( X_t,\hat{X}_t \right)+\lambda \log \frac{P\left(\left. \hat{X}_t,Y^{t-1}\right|\hat{X}^{t-1} \right)}{P\left(\left. \hat{X}_t\right|\hat{X}^{t-1} \right)P\left(Y^{t-1}\left|\hat{X}^{t-1}\right. \right)}\right] 
\end{align}
The next lemma expresses the objective function of \eqref{Eq.OP2} in terms of the belief state $b_t$.
	\begin{lemma} \label{Lm.BeliefStateLoss}
		Let $b_t$ denotes the belief state, \emph{i.e.,} $
			b_t\left( y^{t-1},\tilde{z}^t \right) =  P\left(y^{t-1},\tilde{z}^t\left|\hat{X}^{t-1}\right. \right)$. Given  policy $\mathcal{A}=\left\{\mathcal{A}_t\right\}^T_{t=0}$, we have 
			\begin{align}
		&\mathsf{E}\left[ \left. \textit{l}_{d}\left( X_t,\hat{X}_t \right)	+\lambda \log \frac{P\left(\left. \hat{X}_t,Y^{t-1}\right|\hat{X}^{t-1} \right)}{P\left(\left. \hat{X}_t\right|\hat{X}^{t-1} \right)P\left(Y^{t-1}\left| \hat{X}^{t-1} \right. \right)}\right|\hat{X}^{t-1} \right] \nonumber
			\\
			&=\sum_{y^{t-1},\tilde{z}^t,\hat{x}_t}a_t\left(\left. \hat{x}_t\right|\tilde{z}^t \right) b_t\left( y^{t-1},\tilde{z}^t \right) \Big[ \int{p\left(\left. x_t\right|y^{t-1},\tilde{z}^t \right) \textit{l}_{d}\left( x_t,\hat{x}_t \right) dx_t}\nonumber
			\\
			&+\lambda \log \frac{\sum_{\tilde{z}^t}{a_t\left(\left. \hat{x}_t\right|\tilde{z}^t \right) b_t\left( y^{t-1},\tilde{z}^t \right)}}{\left( \sum_{y^{t-1},\tilde{z}^t}{a_t\left(\left. \hat{x}_t\right|\tilde{z}^t \right) b_t\left( y^{t-1},\tilde{z}^t \right)} \right) \left( \sum_{\tilde{z}^t}{b_t\left( y^{t-1},\tilde{z}^t \right)} \right)} \Big],
		\end{align}
		Moreover, $b_{t+1}$ can be updated as follows
			\begin{equation} \label{Eq.BSUP2}
					b_{t+1}\left( y^t,\tilde{z}^{t+1} \right) =\frac{a _t\left( \left.\hat{X}_t\right|\tilde{z}^t\right) b_t\left( y^{t-1},\tilde{z}^t \right) P\left(\left. y_t\right|y_{t-1} \right) \int{P\left(\left. \tilde{z}_{t+1}\right|x_{t+1} \right) p\left(\left. x_{t+1}\right|y^t,\tilde{z}^t \right) dx_{t+1}}}{\sum_{y^{t-1},\tilde{z}^t}{a_t\left(\left. \hat{X}_t\right|\tilde{z}^t \right) b_t\left( y^{t-1},\tilde{z}^t \right)}}.
			\end{equation}
			We use $b_{t+1}=\varPhi \left( b_t,\mathcal{A}_t,\hat{X}_t \right)$ to compactly represent this recursive update.
	\end{lemma}
\begin{proof}
	See Appendix \ref{App:Lemma3}.
\end{proof}

	At time $T$, given $\hat{X}^{T-1}$, we define the following optimal cost-to-go function,
	\begin{align}
		V_{T}^{\star}\left( \hat{X}^{T-1} \right) =\min_{a_T} \mathsf{E}\left[ \left. \textit{l}_{d}\left( X_T,\hat{X}_T \right)+\lambda \log \frac{P\left(\left. \hat{X}_T,Y^{T-1}\right|\hat{X}^{T-1} \right)}{P\left(\left. \hat{X}_T\right|\hat{X}^{T-1} \right) P\left(Y^{T-1}\left|\hat{X}^{T-1}\right. \right)}\right|\hat{X}^{T-1} \right]
		.\nonumber
	\end{align}
	With Lemma \ref{Lm.BeliefStateLoss}, $V_{T}^{\star}\left( \hat{X}^{T-1} \right)$ is determined by the belief state $b_{T}$, and the optimal policy collection $\mathcal{A}^{\star}_T$. We denote it with $V_{T}^{\star}\left(b_{T} \right)$.\\
	Next, we show by induction that $V_{t}^{\star}\left( \hat{X}^{t-1} \right)$ only depends on $b_{t}$ if $V_{t+1}^{\star}\left( \hat{X}^{t} \right)$ only depends on $b_{t+1}$. Based on the optimality principle, $V_{t}^{\star}\left( \hat{X}^{t-1} \right)$ could be written as,
	\begin{align} \label{Eq.Vt}
			V_{t}^{\star}\left( \hat{X}^{t-1} \right) =&\min_{\mathcal{A}_t} \mathsf{E}\left[\left. \textit{l}_{d}\left( X_t,\hat{X}_t \right)+\lambda \log \frac{P\left(\left. \hat{X}_t,Y^{t-1}\right|\hat{X}^{t-1}\right)}{P\left(\left. \hat{X}_t\right|\hat{X}^{t-1}\right) P\left(Y^{t-1}\left|\hat{X}^{t-1}\right. \right)}\right|\hat{X}^{t-1} \right]\\ &+\mathsf{E}\left[ V_{t+1}^{\star}\left( b_{t+1} \right) \left| \hat{X}^{t-1}\right. \right]  .
	\end{align}
	According to Lemma \ref{Lm.BeliefStateLoss}, the first term is determined by $b_t$. As for the second term, 
	\begin{align}
		\mathsf{E}\left[ V_{t+1}^{\star}\left( b_{t+1} \right) \left| \hat{X}^{t-1}\right. \right] &= \sum_{\hat{x}_t}P\left(\hat{x}_t\left|\hat{X}^{t-1}\right. \right) V_{t+1}^{\star}\left( b_{t+1} \right)\\
		&= \sum_{y^{t-1},\tilde{z}^t,\hat{x}_t}{a_t\left(\left. \hat{x}_t\right|\tilde{z}^t \right) b_t\left( y^{t-1},\tilde{z}^t \right)} V_{t+1}^{\star}\left(\varPhi \left( b_t,\mathcal{A}_t,\hat{x}_t  \right)\right)
	\end{align}
	which is also determined by $b_t$. \\
	\section{Proof of Theorem \ref{Th.PG}}\label{App: Th.PG}
We first derive an expression for the gradient of the mutual information with respect to $\theta$ as follows
	\begin{align}\label{Eq:MI-Grad}
		\nabla _{\theta}I_{\theta}\left( \hat{X}^T;Y^T \right) &\overset{(a)}{=}\nabla _{\theta}\sum_{t=1}^T{I_\theta\left(\left. \hat{X}_t;Y^{t-1}\right|\hat{X}^{t-1} \right)}\nonumber\\
			&\overset{(b)}{=}\mathsf{E}\left[ \sum_{t=0}^T{\log \frac{P_{\theta}\left(\left. \hat{X}_t\right|\hat{X}^{t-1},Y^{t-1} \right) }{P_{\theta}\left(\left. \hat{X}_t\right|\hat{X}^{t-1} \right)}} \nabla _{\theta}\log P_{\theta}\left( \tau \right)\right]+\mathsf{E}\left[ \nabla _{\theta}\sum_{t=0}^T{\log \frac{P_{\theta}\left(\left. \hat{X}_t\right|\hat{X}^{t-1},Y^{t-1} \right)}{P_{\theta}\left(\left. \hat{X}_t\right|\hat{X}^{t-1} \right)}} \right], 
	\end{align}
	where $\paren{a}$ follows from Lemma \ref{Lm.MISimp}, $\paren{b}$ follows from the definition of conditional mutual information and $\tau =\left\{ Y^{T},X^T,\tilde{Z}^T,\hat{X}^T \right\}$. The second term in the right-hand-side of \eqref{Eq:MI-Grad} can be expanded as 
	\begin{align}\label{Eq: log-P-Expan}
	\mathsf{E}\left[ \nabla _{\theta}\log \frac{P_{\theta}\left(\left. \hat{X}_t\right|\hat{X}^{t-1},Y^{t} \right)}{P_{\theta}\left(\left. \hat{X}_t\right|\hat{X}^{t} \right)} \right]={\mathsf{E}\left[ \nabla _{\theta}\log P_{\theta}\left( \left. \hat{X}_t\right|\hat{X}^{t-1},Y^{t-1} \right) \right]}-\mathsf{E}\left[ \nabla _{\theta}\log P_{\theta}\left(\left. \hat{X}_t\right|\hat{X}^{t-1} \right) \right].
	\end{align}

The first term in \eqref{Eq: log-P-Expan} can be expanded as 
\begin{align}
		{\mathsf{E}\left[ \nabla _{\theta}\log P_{\theta}\left(\left. \hat{X}_t\right|\hat{X}^{t-1},Y^{t-1} \right) \right]}
		&=\mathsf{E}\left[ \frac{\nabla _{\theta}P_{\theta}\left(\left. \hat{X}_t\right|\hat{X}^{t-1},Y^{t-1} \right)}{P_{\theta}\left(\left. \hat{X}_t\right|\hat{X}^{t-1},Y^{t-1} \right)} \right] 
		\\
		&=\mathsf{E}\left[ \sum_{\hat{x}_t}{\nabla _{\theta}}P_{\theta}\left(\left. \hat{x}_t\right|\hat{X}^{t-1},Y^{t-1} \right) \right] 
		\\
		&=\nabla _{\theta}\mathsf{E}[1]
		\\
		&=0. \nonumber
\end{align}
Similarly, we have 
\begin{align}
	\mathsf{E}\left[ \nabla _{\theta}\log P_{\theta}\left(\left. \hat{X}_t\right|\hat{X}^{t-1} \right) \right]\nonumber
	&=\mathsf{E}\left[ \frac{\nabla _{\theta}P_{\theta}\left(\left. \hat{X}_t\right|\hat{X}^{t-1} \right)}{P_{\theta}\left(\left. \hat{X}_t\right|\hat{X}^{t-1} \right)} \right] 
	\nonumber\\
	&=\mathsf{E}\left[ \sum_{\hat{x}_t}{\nabla _{\theta}}P_{\theta}\left(\left. \hat{x}_t\right|\hat{X}^{t-1} \right) \right] 
	\nonumber\\
	&=\nabla _{\theta}\mathsf{E}[1]
\nonumber	\\
	&=0.\nonumber
\end{align}
Note that $P_{\theta}\paren{\tau} $ can be written as 
\begin{align}
P_{\theta}\paren{\tau} = \prod_{t=0}^{T}{P\paren{\left. Y_{t} \right| Y_{t-1}}P\paren{\left. X_{t} \right| X_{t-1}, Y_{t-1}}P\paren{\left. \tilde{Z}_{t} \right| X_{t}}\pi_\theta\paren{\left. \hat{X}_{t} \right| h_{t}}},
\nonumber
\end{align}
which implies 
\begin{align}
\nabla _{\theta}\log P_{\theta}\paren{\tau} = \sum_{t=0}^{T}{\nabla _{\theta} \log \pi_\theta\paren{\left. \hat{X}_{t} \right| h_{t}}}.
\nonumber
\end{align}

Thus, we have
		\begin{align}
	\nabla _{\theta}I_{\theta}\left( \hat{X}^T;Y^T \right) =\mathsf{E}\left[ \paren{\sum_{t=0}^T{\log \frac{P_{\theta}\left(\left. \hat{X}_t\right|\hat{X}^{t-1},Y^{t-1} \right) }{P_{\theta}\left(\left. \hat{X}_t\right|\hat{X}^{t-1} \right)}}} \paren{\sum_{t=0}^{T}{\nabla _{\theta} \log \pi_\theta\paren{\left. \hat{X}_{t} \right| h_{t}}}}\right] . \nonumber
	\end{align}

Note that the gradient of the estimation error with respect to $\theta$ can be written as 
\begin{align}
\nabla _{\theta}\mathsf{E}\left[ \sum_{t=0}^T{\textit{l}_{d}\left( X_t,\hat{X}_t \right)} \right] =\mathsf{E}\left[ \paren{\sum_{t=0}^T{\textit{l}_{d}\left( X_t,\hat{X}_t \right)}} \paren{\sum_{t=0}^{T}{\nabla _{\theta} \log \pi_\theta\paren{\left. \hat{X}_{t} \right| h_{t}}}} \right],\nonumber
\end{align}
thus we complete the proof.

	\section{Proof of Theorem \ref{Th.ILA}}
	\begin{lemma} \label{Lm.KLV}
		Given two distributions $P\left(\cdot\right)$ and $Q\left(\cdot\right)$ with the same support $\mathcal{S}$, 	the KL-divergence between distributions $P\left(\cdot\right)$ and $Q\left(\cdot\right)$ has the following variational formulation \cite{nguyen2010estimating,poole2018variational},
		\begin{equation}
			D_{KL}\left[ \left.P\left(S\right)\right\|Q\left(S\right)\right] =\mathop {\sup} \limits_{f}{\mathsf{E}_P\left[ f\left(S\right) \right] -e^{-1}\mathsf{E}_Q\left[ e^{f\left(S\right)} \right]},
		\end{equation}
		where $f$ is a measurable function, with optimal solution
		$$f^\star\left(s\right)=\log \frac{P\left(s\right)}{Q\left(s\right)}+1 .$$
	\end{lemma}
	\begin{align}
			I_\theta\left( \hat{X}^T;Y^T \right) & =\sum_{t=0}^T{I_\theta\left(\left. \hat{X}_t;Y^t\right|\hat{X}^{t-1} \right)}
			\\
			&\overset{(a)}{=}\sum_{t=0}^T{\mathsf{E}\left[ \log \frac{P_\theta\left(\left. \hat{X}_t\right|\hat{X}^{t-1},Y^t \right)}{P\left( \tilde{X} \right)} \right] -\mathsf{E}\left[ \log \frac{P_\theta\left(\left. \hat{X}_t\right|\hat{X}^{t-1} \right)}{P\left( \tilde{X} \right)} \right]}
			\\
			&=\sum_{t=0}^T{\mathsf{E}\left[ \log \frac{P_\theta\left( \hat{X}^t,Y^t \right)}{P\left( \tilde{X} \right) P_\theta\left( \hat{X}^{t-1},Y^t \right)} \right] -\mathsf{E}\left[ \log \frac{P_\theta\left( \hat{X}^t \right)}{P\left( \tilde{X} \right) P_\theta\left( \hat{X}^{t-1} \right)} \right]}
			\\
			&=\sum_{t=0}^T{D_{KL}\left[ \left. P_\theta\left(\hat{X}^t,Y^t\right)\right\|P_\theta\left(\hat{X}^{t-1},Y^t\right)P\left(\tilde{X}\right) \right] -D_{KL}\left[\left. P_\theta\left({\hat{X}^t}\right)\right\|P_\theta\left(\hat{X}^{t-1}\right)P\left(\tilde{X}\right) \right]}, \nonumber
	\end{align}
	where we impose a discrete uniform variable $\tilde{X}$ in $\paren{a}$ to form two KL-divergence. For each KL-divergence, we have the following variational formulations based on the Lemma \ref{Lm.KLV},
	\begin{align}
			D_{KL}\left[ \left. P_\theta\left(\hat{X}^t,Y^t\right)\right\|P_\theta\left(\hat{X}^{t-1},Y^t\right)P\left(\tilde{X}\right) \right] =\mathop {\mathrm{sup}} \limits_{g_t}\mathsf{E}\left[  g_t\left( \hat{X}_t,\hat{X}^{t-1},Y^t \right) -e^{g_t\left( \tilde{X},\hat{X}^{t-1},Y^t \right) -1} \right] , \nonumber
	\end{align}
	\begin{align}
			D_{KL}\left[\left. P_\theta\left({\hat{X}^t}\right)\right\|P_\theta\left(\hat{X}^{t-1}\right)P\left(\tilde{X}\right) \right] =\mathop {\mathrm{sup}} \limits_{f_t}\mathsf{E}\left[ f_t\left( \hat{X}_t,\hat{X}^{t-1} \right) - e^{f_t\left( \tilde{X},\hat{X}^{t-1} \right) -1}  \right] . \nonumber
	\end{align}
	By optimizing functions over the whole horizon, we have that,
	\begin{align}
			\sum_{t=0}^T{D_{KL}\left[\left. P_\theta\left(\hat{X}^t,Y^t\right)\right\|P_\theta\left(\hat{X}^{t-1},Y^t\right)P\left(\tilde{X}\right) \right] }=\mathop {\mathrm{sup}} \limits_{\{g_t\}_{t=0}^T}\mathsf{E}\left[ \sum_{t=0}^T{g_t\left( \hat{X}_t,\hat{X}^{t-1},Y^t \right) -e^{g_t\left( \tilde{X},\hat{X}^{t-1},Y^t \right) -1}} \right] , \nonumber
	\end{align}
	\begin{align}
			\sum_{t=0}^T{D_{KL}\left[\left. P_\theta\left({\hat{X}^t}\right)\right\|P_\theta\left(\hat{X}^{t-1}\right)P\left(\tilde{X}\right) \right]}=\mathop {\mathrm{sup}} \limits_{\{f_t\}_{t=0}^T}\mathsf{E}\left[ \sum_{t=0}^T{ f_t\left( \hat{X}_t,\hat{X}^{t-1} \right) - e^{f_t\left( \tilde{X},\hat{X}^{t-1} \right) -1}} \right], \nonumber
	\end{align}
	where the optimal functions are,
	$$g_t^\star\left( \hat{x}_t,\hat{x}^{t-1},y^t \right) =\log \frac{P_\theta\left(\left. \hat{x}_t\right|\hat{x}^{t-1},y^t \right)}{P\left( \tilde{x} \right)}+1,$$
	$$f_t^\star\left( \hat{x}_t,\hat{x}^{t-1} \right) =\log \frac{P_\theta\left(\left. \hat{x}_t\right|\hat{x}^{t-1} \right)}{P\left( \tilde{x} \right)}+1.$$
	Thus, the information loss could be obtained by,
	$$\log \frac{P_\theta\left(\left. \hat{x}_t\right|\hat{x}^{t-1},y^t \right)}{P_\theta\left(\left. \hat{x}_t\right|\hat{x}^{t-1} \right)}=g_t^\star\left( \hat{x}_t,\hat{x}^{t-1},y^t \right) -f_t^\star\left( \hat{x}_t,\hat{x}^{t-1} \right) .$$

\section{Proof of Lemma \ref{Lm.MISimp}}\label{App: MIE}
The mutual information between $\hat{X}^T$ and $Y^T$ can be expanded as 
\begin{align}
		I\left( \hat{X}^T;Y^T \right) & \overset{(a)}{=}\sum_{t=0}^T{I\left(\left. \hat{X}_t;Y^T\right|\hat{X}^{t-1} \right)}\\
		&\overset{(b)}{=}\sum_{t=0}^T{I\left(\left. \hat{X}_t;Y^{t-1}\right|\hat{X}^{t-1} \right)+I\left(\hat{X}_t;Y^{T}_{t}|\hat{X}^{t-1},Y^{t-1} \right)}\\
		&\overset{(c)}{=}\sum_{t=0}^T{I\left(\left. \hat{X}_t;Y^{t-1}\right|\hat{X}^{t-1} \right)},
		   \nonumber
\end{align}
where $\paren{a}$ and $\paren{b}$ follows from the chain rule, $\paren{c}$ follows from the fact that the $Y_t^T$ is independent from $\hat{X}^t$ given the $Y^{t-1}$, i.e., the Markov chain $\hat{X}_t \rightarrow Y^{t-1} \rightarrow Y_t^{T}$.\\

\section{Proof of Lemma \ref{Lm.BeliefStateLoss}}\label{App:Lemma3}
Using the definition of $b_t$, we have 
\begin{equation} 
	\begin{split}
		b_{t+1}\left( y^t,\tilde{z}^{t+1} \right) &=\frac{P\left(\left. y^t,\tilde{z}^{t+1},\hat{X}_t\right|\hat{X}^{t-1} \right)}{P\left(\left. \hat{X}_t\right|\hat{X}^{t-1} \right)}
		\\
		&\overset{(a)}{=}\frac{\int{P\left(\left. \tilde{z}_{t+1}\right|x_{t+1},y^t,\tilde{z}^t,\hat{X}^t \right) p\left(\left. x_{t+1}\right|y^t,\tilde{z}^t,\hat{X}^t \right)a_t\left(\left. \hat{X}_t\right|\tilde{z}^t \right) b_t\left( y^{t-1},\tilde{z}^t \right) P\left(\left. y_t\right|y_{t-1} \right)  dx_{t+1}}}{\sum_{y^{t-1},\tilde{z}^t}{a_t\left(\left. \hat{X}_t\right|\tilde{z}^t \right) b_t\left( y^{t-1},\tilde{z}^t \right)}}
		\\
		&\overset{(b)}{=}\frac{a_t\left(\left. \hat{X}_t\right|\tilde{z}^t \right) b_t\left( y^{t-1},\tilde{z}^t \right) P\left(\left. y_t\right|y_{t-1} \right) \int{P\left(\left. \tilde{z}_{t+1}\right|x_{t+1} \right) p\left(\left. x_{t+1}\right|y^t,\tilde{z}^t \right) dx_{t+1}}}{\sum_{y^{t-1},\tilde{z}^t}{a_t\left(\left. \hat{X}_t\right|\tilde{z}^t \right) b_t\left( y^{t-1},\tilde{z}^t \right)}}
		,
	\end{split}
\end{equation}
where the numerator in $\paren{a}$ is derived by expanding the $P\left(\left. y^t,\tilde{z}^{t+1},\hat{X}_t\right|\hat{X}^{t-1} \right)$ as the multiplication of the system dynamics, policy and $b_t$, the denominator in $\paren{a}$ shows that $P\left( \left. \hat{X}_t\right|\hat{X}^{t-1}\right)$ is a linear combination of $\pi _t\left(\left. \hat{X}_t\right|\tilde{z}^t,\hat{X}^{t-1} \right)$ and  $b_t\left( y^{t-1},\tilde{z}^t \right)$. As for $\paren{b}$, $P\left(\left. \tilde{z}_{t+1}\right|x_{t+1},y^t,\tilde{z}^t,\hat{X}^t \right)=P\left(\left. \tilde{z}_{t+1}\right|x_{t+1}\right)$ is due to the observation probability, and $p\left(\left. x_{t+1}\right|y^t,\tilde{z}^t,\hat{X}^t \right)=p\left(\left. x_{t+1}\right|y^t,\tilde{z}^t \right)$ is due to the Markov chain $\hat{X}^t\rightarrow \left(\tilde{Z}^t,Y^{t}\right) \rightarrow X_{t+1}$.\\
For the conditional distortion, we have that
\begin{align}
		&\mathsf{E}\left[ \left. l_{d}\left( X_t,\hat{X}_t \right) \right|\hat{X}^{t-1} \right] 
		\\
		&=\sum_{y^{t-1},\tilde{z}^t,\hat{x}_t}{\int{p\left(\left. x_t\right|y^{t-1},\tilde{z}^t,\hat{X}^t,a ^{t-1} \right) \pi _t\left(\left. \hat{x}_t\right|\tilde{z}^t,\hat{X}^{t-1} \right) P\left(\left. y^{t-1},\tilde{z}^t\right|\hat{X}^{t-1},a ^{t-1} \right) l_{d}\left( x_t,\hat{x}_t \right) dx_t}}
		\\
		&=\sum_{y^{t-1},\tilde{z}^t,\hat{x}_t}{\int{p\left(\left. x_t\right|y^{t-1},\tilde{z}^t \right) a_t\left(\left. \hat{x}_t\right|\tilde{z}^t \right) b_t\left( y^{t-1},\tilde{z}^t \right) l_{d}\left( x_t,\hat{x}_t \right) dx_t}}.\nonumber
\end{align}
As for the conditional information loss, 
\begin{align}
		&\mathsf{E}\left[ \left. \log \frac{P\left(\left. \hat{X}_t,Y^{t-1}\right|\hat{X}^{t-1} \right)}{P\left(\left. \hat{X}_t\right|\hat{X}^{t-1} \right) P\left(\left. Y^{t-1}\right|\hat{X}^{t-1} \right)} \right|\hat{X}^{t-1} \right] 
		\\
		&=\sum_{y^{t-1},\hat{x}_t}{P\left(\left. \hat{x}_t,y^{t-1}\right|\hat{x}^{t-1} \right) \log \frac{P\left(\left. \hat{x}_t,y^{t-1}\right|\hat{x}^{t-1} \right)}{P\left(\left. \hat{x}_t\right|\hat{x}^{t-1} \right) P\left(\left. y^{t-1}\right|\hat{x}^{t-1} \right)}},
\end{align}
each probability inside satisfies
$$P\left(\left. \hat{x}_t,y^{t-1}\right|\hat{x}^{t-1} \right) =\sum_{\tilde{z}^t}{a_t\left(\left. \hat{x}_t\right|\tilde{z}^t,\hat{x}^{t-1} \right) P\left(\left. y^{t-1},\tilde{z}^t\right|\hat{x}^{t-1} \right)}=\sum_{\tilde{z}^t}{a_t\left( \left. \hat{x}_t\right|\tilde{z}^t \right) b_t\left( y^{t-1},\tilde{z}^t \right)},$$

$$P\left(\left. \hat{x}_t\right|\hat{x}^{t-1} \right) =\sum_{y^{t-1},\tilde{z}^t}{a_t\left(\left. \hat{x}_t\right|\tilde{z}^t,\hat{x}^{t-1} \right) P\left(\left. y^{t-1},\tilde{z}^t\right|\hat{x}^{t-1} \right)}=\sum_{y^{t-1},\tilde{z}^t}{a_t\left( \left. \hat{x}_t\right|\tilde{z}^t \right) b_t\left( y^{t-1},\tilde{z}^t \right)},$$

$$P\left(\left. y^{t-1}\right|\hat{x}^{t-1} \right) =\sum_{\tilde{z}^t}{P\left(\left. y^{t-1},\tilde{z}^t\right|\hat{x}^{t-1} \right)}=\sum_{\tilde{z}^t}{b_t\left( y^{t-1},\tilde{z}^t \right)}.$$
Therefore,
\begin{align}
		&\mathsf{E}\left[ \left. \textit{l}_{d}\left( X_t,\hat{X}_t \right)+\lambda \log \frac{P\left(\left. \hat{X}_t,Y^{t-1}\right|\hat{X}^{t-1} \right)}{P\left(\left. \hat{X}_t\right|\hat{X}^{t-1} \right)P\left(Y^{t-1}\left|\hat{X}^{t-1}\right. \right)}\right|\hat{X}^{t-1} \right]
		\\&=\sum_{y^{t-1},\tilde{z}^t,\hat{x}_t}a_t\left(\left. \hat{x}_t\right|\tilde{z}^t \right) b_t\left( y^{t-1},\tilde{z}^t \right) \Big[ \int{p\left(\left. x_t\right|y^{t-1},\tilde{z}^t \right) \textit{l}_{d}\left( x_t,\hat{x}_t \right) dx_t}
		\\
		&\quad\quad+\lambda \log \frac{\sum_{\tilde{z}^t}{a_t\left(\left. \hat{x}_t\right|\tilde{z}^t \right) b_t\left( y^{t-1},\tilde{z}^t \right)}}{\left( \sum_{y^{t-1},\tilde{z}^t}{a_t\left(\left. \hat{x}_t\right|\tilde{z}^t \right) b_t\left( y^{t-1},\tilde{z}^t \right)} \right) \left( \sum_{\tilde{z}^t}{b_t\left( y^{t-1},\tilde{z}^t \right)} \right)} \Big]. \nonumber
\end{align}

\end{document}